\documentclass[conference]{IEEEtran}
\IEEEoverridecommandlockouts
\usepackage[square,numbers,sort&compress]{natbib}
\usepackage{amsmath,amssymb,amsfonts}
\usepackage{amsthm}
\usepackage{algorithmic}
\usepackage{graphicx}
\usepackage{textcomp}
\usepackage{xcolor}
\usepackage{complexity}
\usepackage{authblk}
\usepackage{fancyhdr}
\def\BibTeX{{\rm B\kern-.05em{\sc i\kern-.025em b}\kern-.08em
    T\kern-.1667em\lower.7ex\hbox{E}\kern-.125emX}}
\usepackage{bm}
\usepackage{physics}
\usepackage{qcircuit}
\usepackage{tikz}
\usepackage{caption}
\usepackage{subcaption}
\usepackage[export]{adjustbox}
\newcommand*\circled[1]{\tikz[baseline=(char.base)]{
            \node[shape=circle,draw,inner sep=1pt] (char) {#1};}}
            
\newcommand*\circledr[1]{\tikz[baseline=(char.base)]{
            \node[red,shape=circle,draw,inner sep=1pt] (char) {#1};}}
            
\newcommand*\circledb[1]{\tikz[baseline=(char.base)]{
            \node[blue,shape=circle,draw,inner sep=1pt] (char) {#1};}}

\newtheorem{theorem}{Theorem}[section]

\newtheorem{lemma}[theorem]{Lemma}

\newcommand{\sdotfill}{\textcolor[rgb]{0.8,0.8,0.8}{\dotfill}} %change to \cdotfill later

\renewcommand{\vec}[1]{\boldsymbol{#1}}

\author[1,2]{Wei Zi}
\author[1]{Qian Li}
\author[1,3, *]{Xiaoming Sun}
\affil[1]{State Key Lab of Processors, Institute of Computing Technology, Chinese Academy of Sciences, Beijing 100190, China}
\affil[2]{University of Chinese Academy of Sciences, Beijing 100049, China}
\affil[3]{CAS Center for Excellence in Topological Quantum Computation, University of Chinese Academy of Sciences}
\affil[*]{sunxiaoming@ict.ac.cn}

\begin{document}

\title{Optimal Synthesis of Multi-Controlled Qudit Gates
%\thanks{Wei Zi and Qian Li contributed to this work equally; Xiaoming Sun is the corresponding author.}
}

\maketitle
\begin{abstract}

%%%%%%%%%%%%%%%%%%%%%DAC version abstract%%%%%%%%%
We propose a linear-size synthesis of the multi-controlled Toffoli gate on qudits with at most one borrowed ancilla. This one ancilla can even be saved when the qudit dimension is odd. Our synthesis leads to improvements in various quantum algorithms implemented on qudits. In particular, we obtain (i) a linear-size and one-clean-ancilla synthesis of multi-controlled qudit gates; (ii) an optimal-size and one-clean-ancilla synthesis of unitaries on qudits; (iii) a near-optimal-size and ancilla-free/one-borrowed-ancilla implementation of classical reversible functions as qudit gates.

\end{abstract}

\section{Introduction}
At present, quantum computing has entered the era of noisy medium-scale quantum (NISQ) systems, where there are inherent limitations on size, depth, and the number of qubits of quantum circuits that can be supported by physical experimental hardware, so the degree of optimization of quantum circuits directly affects the scope of application of quantum computers \cite{preskill2018quantum}. Designing quantum circuits as small as possible, as shallow as possible, and using as few qubits as possible for various computational problems is one of the most important research directions in the field of quantum computing \cite{shende2005synthesis,hu2019efficient,increamenter1,bullock2005asymptotically,soeken2016unlocking,bhattacharjee2019muqut,chu2023scalable,sun2021asymptotically}.

While typical quantum circuits are expressed in terms of qubits (two-level quantum systems), many of the underlying physical systems, e.g. the quantum processors based on photonic systems \cite{lu2020quantum,chi2022programmable}, ion traps \cite{klimov2003qutrit,ringbauer2022universal}, and superconducting devices \cite{blok2021quantum,yurtalan2020implementation}, have much higher natural dimensions, where the proposed type of qubit is actually a restricted subspace of the higher-dimensional systems. By utilizing the other wasted dimensions accessible, mainly because of the increment of the device's information density, we can reduce the resource requirements of quantum circuits \cite{campbell2014enhanced,wang2020qudits}, even exponentially \cite{increamenter1}! So, to extend the frontier of what quantum computers can compute, in particular for those whose underlying physical system (e.g., those based on ion traps \cite{bruzewicz2019trapped}) suffers from poor scalability (i.e., supports only a small number of qubits), it is a promising way to utilize the higher dimensions and work with qudits instead of qubits to implement quantum circuits.

Compared with qubit circuits, there is much less research on the optimization of qudit circuits, and the qudit synthesis of many basic primitives remains to be optimized. The class of multi-controlled qudit gates is an important such primitive, which is widely used in many quantum algorithms, including unitary synthesis \cite{bullock2005asymptotically,toffoli-sy0,toffoli-sy1}, Grover's search algorithm \cite{saha2022asymptotically}, arithmetic operators synthesis \cite{adder1,adder2}, and implementation of classical reversible functions \cite{yeh2022constructing}. There is a standard synthesis of multi-controlled $d$-level qudit gates by using $O(k)$ two-qudit gates, whose two-qudit gate count is optimal, but using as many as $\lceil (k-2)/(d-2) \rceil$ clean ancilla \cite{bullock2005asymptotically,adder2}. Here, $k$ is the number of controls. The synthesis in \cite{moraga2016quantum} is ancilla-free but uses an exponential number of two-qudit gates.
Di and Wei \cite{toffoli-sy1} claimed an ancilla-free synthesis by using $O(k^3)$ two-qudit gates, which significantly improves the synthesis in \cite{moraga2016quantum}. 
Recently, Yeh and van de Wetering \cite{yeh2022constructing} studied how to synthesize multi-controlled qutrit ($3$-level qudit) gates in a fault-tolerant manner, and obtained an ancilla-free synthesis of any multi-controlled Clifford+T unitary on qutrits by using $O(k^{3.585})$ Clifford+T gates. 

In this paper, we propose a one-clean-ancilla synthesis of any $k$-controlled qudit gate by using just $O(k)$ two-qudit gates, which achieves optimality both on size and number of ancilla up to just one ancilla. The core of our synthesis is an $O(k)$-size synthesis of a special $k$-controlled qudit gate, namely the $k$-controlled Toffoli gate using no ancilla when $d$ is odd or just one borrowed ancilla when $d$ is even. In addition, our synthesis of the $k$-controlled Toffoli gate directly leads to an improvement from $O(k^{3.585})$ to $O(k)$ in the Clifford+T gate count of \cite{yeh2022constructing}'s synthesis mentioned above. As applications, our synthesis can be used to improve various quantum algorithms, e.g., synthesis of arithmetic operators \cite{adder1,adder2} and $d$-ary Grover's algorithm \cite{saha2022asymptotically}. In particular, it has the following significant implications.
\vspace{0.25ex}

\noindent\underline{Unitary Synthesis.}  Bullock et al. \cite{bullock2005asymptotically} showed that any unitary on $n$ $d$-level qudits can be synthesized by using $O(d^{2n})$ two-qudit gates, which has been shown to be optimal \cite{bullock2005asymptotically}, but by using $\lceil (n-2)/(d-2) \rceil$ clean ancilla.

	Here, by substituting our improved synthesis of multi-controlled qudit gates, we can significantly reduce the number of clean ancilla from $\lceil (n-2)/(d-2) \rceil$ to just $1$, while keeping the two-qudit gate count still optimal (see Section \ref{sec:unitary}).
	\vspace{0.25ex}

\noindent\underline{Implementation of Classical Reversible Functions.} A $n$-variable $d$-ary classical reversible function is a bijective map $f:\{0,1,\dots,d-1\}^n \to \{0,1,\dots,d-1\}^n$. Classical reversible functions are important because of not only the energy-efficiency of reversible logic but also quantum algorithms involving oracles, which implement classical functions using quantum gates. \cite{yeh2022constructing} obtained an $O(3^nn^{3.585})$-size ancilla-free implementation of any $n$-variable ternary classical reversible function in a fault-tolerant manner. 

Here, for any $d\geq 3$, by substituting our improved synthesis of multiple-controlled Toffoli gate, we obtain an $O(d^nn)$-size implementation of any $n$-variable $d$-ary classical reversible function, where the size is optimal up to a logarithmic factor, and using no ancilla when $d$ is odd and just one borrowed ancilla when $d$ is even (see Section \ref{sec:reversible}). In particular, when $d=3$, our implementation remains fault-tolerant, which answers an open question proposed in \cite{yeh2022constructing}. 
\vspace{1ex}

The rest of this paper is organized as follows. Section~\ref{sec:pre} presents preliminaries. In Section~\ref{sec:main}, we show how to synthesize multiple-controlled qudit gates. In Section~\ref{sec:app}, we apply our synthesis to improve the unitary synthesis and the implementation of classical reversible functions. We conclude this paper in Section~\ref{sec:con}.

\section{Preliminaries}\label{sec:pre}
For a positive integer $d$, let $[\underline{d}]$ denote the set $\{0,1,\cdots,d-1\}$. We will
use boldface type characters, e.g., $\vec{x}$, for vectors. For a vector $\vec{x}=(x_1,\cdots,x_n)$ and $1\leq i<j\leq n$, let $\boldsymbol{x}_{i:j}$ denote the subvector $(x_i, x_{i+1}, \dots, x_j)$.

A qubit is a two-level quantum-mechanical system, or mathematically associated with a two-dimensional Hilbert space.
Similarly, a $d$-level qudit is associated with a $d$-dimensional Hilbert space where $d\geq 3$ is an integer. Let $\ket{0},\ket{1},\cdots,\ket{d-1}$ denote the computational basis of a $d$-level qudit. Any state on a $d$-level qudit can be written
as $\ket{\phi}=\sum_{i=0}^{d-1} \alpha_{i}\ket{i}$ where each $\alpha_i$ is a complex number and $\sum_{i=0}^{d-1}|\alpha_i|^2=1$, or mathematically is a unit vector in the Hilbert space. Throughout the paper, we treat $d$ as a constant, and a $\poly(d)$ factor may be hidden in the big $O$ notation.

Quantum states can be acted on by quantum gates, which are mathematically unitary operators on the Hilbert space. We introduce some quantum gates acting on qudits that we will meet. 
\vspace{0.5ex}

\noindent\underline{Single-qudit gates.} For two distinct $i,j\in[\underline{d}]$, the $X_{ij}$ gate, which acts on a $d$-level qudit, swaps $\ket{i}$ and $\ket{j}$ and leaves the other computational basis unchanged. For example, applying $X_{01}$ to the state $\ket{\phi}=\alpha_0\ket{0}+\alpha_1\ket{1}+\sum_{i=2}^{d-1}\alpha_{i}\ket{i}$ produces $X_{01}\ket{\phi}=\alpha_1\ket{0}+\alpha_0\ket{1}+\sum_{i=2}^{d-1}\alpha_{i}\ket{i}$. For integer $y$, the $X_{+y}$ gate sends $\ket{i}$ to $\ket{(i+y)\mod d}$ for each $i\in[\underline{d}]$. Because any permutation can be decomposed into a product of at most $d-1$ swap operations \cite{dixon1996permutation}, $X_{+y}$ can be synthesized by at most $d-1$ $X_{ij}$ gates.
\vspace{0.5ex}

\noindent\underline{Controlled gates.} Let $U$ be a single-qudit gate acting on a $d$-level qudit, which is mathematically a $d\times d$ unitary. The $\ket{0}$-controlled $U$ (or $\ket{0}$-$U$ for short), which is a two-qudit gate, acts as 
\[
\ket{0}\otimes\ket{\phi}\mapsto \ket{0}\otimes U\ket{\phi},\quad \ket{i}\otimes\ket{\phi}\mapsto \ket{i}\otimes \ket{\phi} \text{for }i\neq 0. 
\]
That is, it implements $U$ on the target qudit if and only if the control qudit is in the $\ket{0}$ state.  The circuit representation for the $\ket{0}$-$U$ gate is shown in Fig.~\ref{fig:CU_CnU}(a), where the top line and bottom represent the control qudit and target qudit respectively. Let $X_{ij}$ and $X_{+y}$ instantiating $U$ respectively, then we get the $\ket{0}$-$X_{ij}$ and $\ket{0}$-$X_{+y}$ gates. Similarly, we can define $\ket{\ell}$-$U$ for $\ell\in[\underline{d}]$, which fires only when the control qudit is in the $\ket{\ell}$ state. Moreover, we let $\ket{o}$-$U:=\Pi_{\text{odd } \ell}(\ket{\ell}$-$U)$ (and $\ket{e}$-$U:=\Pi_{\text{even } \ell, \ell \ne 0}(\ket{\ell}$-$U)$ resp.) denote the gate that implements $U$ when the control qudit is in the odd (and non-zero even resp.) computational basis.

Let $\mathcal{G}$ denote the gate set $\{\ket{0}\text{-}X_{01}\}\cup\{X_{ij}:i\neq j\}$, and call gates from $\mathcal{G}$ $\mathcal{G}$-gates. An easy observation is that both $\ket{\ell}$-$X_{+y}$ and $\ket{\ell}$-$X_{ij}$ can be synthesized by using $O(d)$  $\mathcal{G}$-gates.
\vspace{0.5ex}

\noindent\underline{Multi-controlled gates.} A multi-controlled gate is just adding more control qudits to a controlled gate. Specifically, the $\ket{0^k}$-controlled $U$ (or $\ket{0^k}$-$U$ for short), where there are $k$ control qudits and one target qudit, acts as 
\[
\ket{0^k}\otimes\ket{\phi}\! \mapsto \! \ket{0^k}\otimes U\ket{\phi}, \ket{\vec{x}}\otimes\ket{\phi}\! \mapsto \! \ket{\vec{x}}\otimes \ket{\phi} \text{for }\vec{x} \in[\underline{d}]^k \backslash 0^k. 
\]
Let $X_{ij}$ and $X_{+y}$ instantiating $U$ respectively, the we get the $\ket{0^k}$-$X_{ij}$ and $\ket{0^k}$-$X_{+y}$ gates. In addition, we also call the $\ket{0^k}$-$X_{01}$ gate the $k$-Toffoli gate. 

For a gate $U$, $U^{\dagger}$ is the inverse of $U$, or mathematically the adjoint of $U$. In particular, $U^{\dagger} = U$ for each $U\in\mathcal{G}$, $X_{+y}^{\dagger} = X_{+(d-y)}$, and $(U_{n}U_{n-1}\cdots U_1)^{\dagger}=U_1^{\dagger} \cdots U_{n-1}^{\dagger} U_n^{\dagger}$. Recall that $UU^\dagger=U^\dagger U=I$. Here, $I$ is the identity operator. 

\begin{figure}
    \centering
    \includegraphics[scale=0.24]{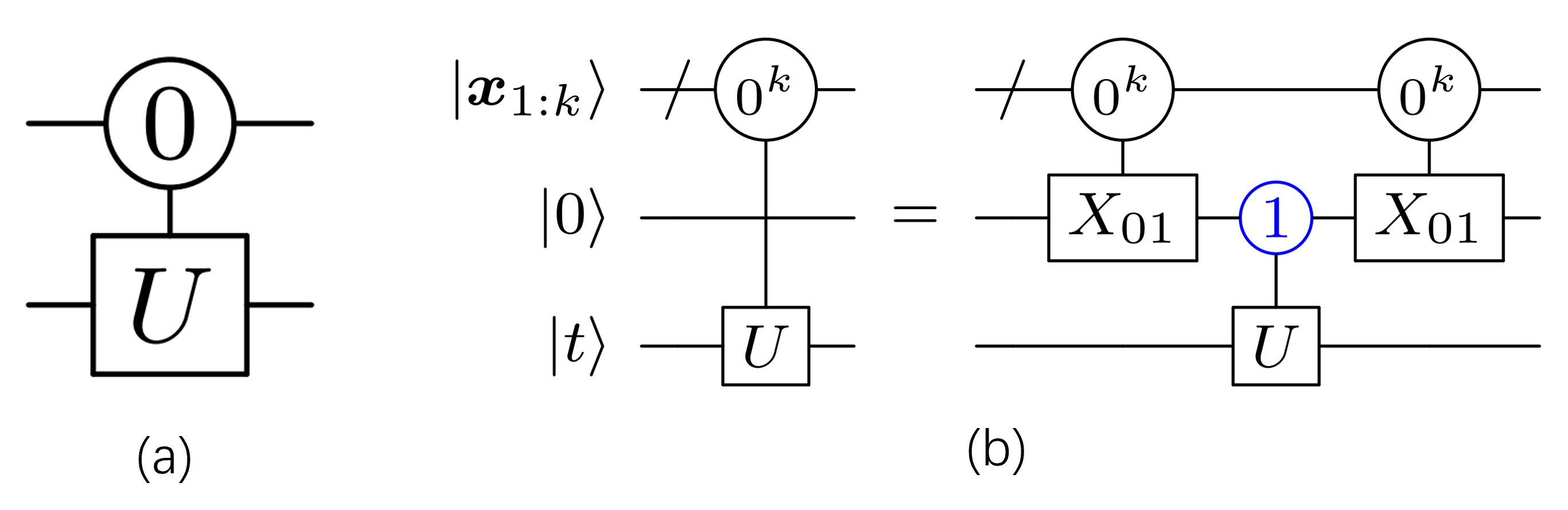}
    \caption{(a) The circuit of $\ket{0}$-$U$. (b) The synthesis of $\ket{0^k}$-$U$.}
    \label{fig:CU_CnU}
\end{figure}

Ancilla qudits are extra qudits not involved in the logical operation that is performed.
According to the initial state and final state, ancilla qudits can be classified into four types:
\begin{itemize}
	\item A {\it Burnable} Ancilla is an ancilla whose initial state is $\ket{0}$ and final state can be arbitrary.
	\item A {\it Clean} Ancilla is an ancilla whose initial state and final state are both  $\ket{0}$.
	\item A {\it Garbage} Ancilla is an ancilla whose initial state and the final state can be both arbitrary.
	\item A {\it Borrowed} Ancilla is an ancilla whose initial state can be arbitrary and final state is the same as the initial state.
\end{itemize}

\section{Synthesis of Multi-Controlled Gates}
\label{sec:main}

In this section, we show how to synthesize the $\ket{0^k}$-$U$ gate by using $O(k)$ two-qudit gates and one clean ancilla. The core is a synthesis of the $k$-Toffoli gate by using $O(k)$ $\mathcal{G}$-gates and at most one borrowed ancilla (see Theorem \ref{the:X01_even} and \ref{the:X01}), which directly leads to the desired synthesis of $\ket{0^k}$-$U$ as shown in Fig.~\ref{fig:CU_CnU}(b). 

In the rest of this section, we focus on the synthesis of the $k$-Toffoli gate on $d$-level qudits. When $d$ is even, the synthesis is essentially the same to that for qubits \cite{barenco1995elementary}. When $d$ is odd, the synthesis turns out to be totally different, which is the main technical part of this paper. 
\vspace{0.5ex}

\subsection{Synthesis of the $k$-Toffoli gate when $d$ is even}

For even $d$, we define $X_{eo}^e:=X_{01}X_{23}\cdots X_{(d-2)(d-1)}$, which swaps the even computational basis with odd ones. Here, the ``$e$" and ``$o$" in the subscript mean ``even" and ``odd", and the ``$e$" in the superscript is to distinguish it from $X_{eo}^o$ which will be defined in Section \ref{subsection:32}.

\begin{lemma}
    \label{lem:2toffoli_even}
	For even $d\geq 3$, the $\ket{00}$-$X_{01}$ can be synthesized by using $O(d)$ $\mathcal{G}$-gates and one borrowed ancilla.
\end{lemma}
\begin{proof}
    Fig.~\ref{fig:CCX_evend} presents a synthesis of $2$-Toffoli by using $O(d)$ $\mathcal{G}$-gates and one borrowed ancilla. To verify the correctness, we will show that: after implementing the circuit, (i) the controls $\ket{\boldsymbol{x}_{1:2}}$ remain unchanged; (ii) the target $\ket{t}$ becomes $X_{01}\ket{t}$ if $x_1=x_2=0$ and unchanged otherwise.

    {\it Part (i)}. After removing the two $\ket{0}$-$X_{01}$ gates which targeted $\ket{t}$, the remaining gates pair off and cancel each other out in a one-to-one manner centered around the deleted gates. For instance, the gates on the left and right of the first removed gate cancel out, followed by the gates on the left and right of the previously eliminated gates, and so on. The outcome is an empty circuit, which confirms that the control qudits remain unaltered after executing the whole circuit. %throughout the entire process. 
    {Besides, it is worth noting that the control qudits remain unchanged just after executing the circuit on the left side of the dashed vertical line.}  %It is worth noting that this also holds true after executing the circuit on both sides of the barrier.

    {\it Part (ii)}. Firstly, we compute and list all the possible input strings that could activate the first $\ket{0}$-$X_{01}$ gate targeting $\ket{t}$:
    \begin{itemize}
    \item When $x_1 = 0$ and $x_2 \notin \{0,1\}$.
    \item When $x_1 = 0$, $x_2 = 0$ or $1$, and $a$ is even.
    \item When $x_1 = 1$, $x_2 = 0$ or $1$, and $a$ is odd.
    \end{itemize}
    {Recalling that the control qudits remain unchanged just after executing the circuit on the left side of the dashed vertical line, one can easily list all the possible input strings that could activate the second $\ket{0}$-$X_{01}$ gate targeting $\ket{t}$:}
%    Note that since the left part circuit leaves the control qudits unaltered, we can focus on the circuit starting at the barrier for analysis. Next, we compute and list all the possible input strings that could activate the second $\ket{0}$-$X_{01}$ gate targeting $\ket{t}$:
    \begin{itemize}
    \item When $x_1 = 0$ and $x_2 \notin \{0,1\}$.
    \item When $x_1 = 0$, $x_2 = 0$, and $a$ is odd.
    \item When $x_1 = 0$, $x_2 = 1$, and $a$ is even.
    \item When $x_1 = 1$, $x_2 = 0$ or $1$, and $a$ is odd.
    \end{itemize}
    Thus, only when the input string satisfies $x_1 = x_2 = 0$, regardless of the state of $\ket{a}$, {exactly one of the two $X_{01}$'s }%one $X_{01}$ gate 
    is applied to the target qudit $\ket{t}$. {The conclusion is now immediate by noting that $X_{01}^2 = I$.}
\end{proof}

The 2-Toffoli gate will be used as a gadget to synthesize $k$-Toffoli for larger $k$.

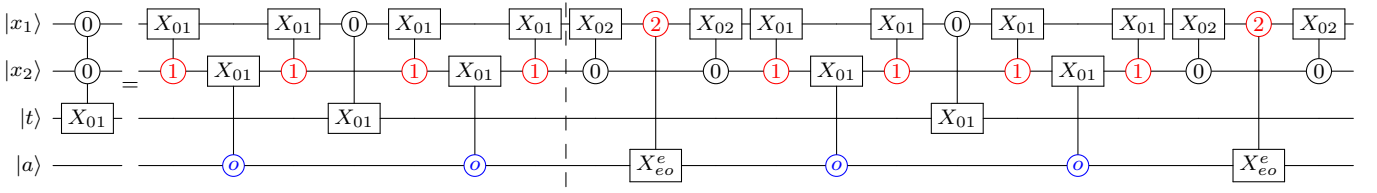
\begin{figure*}[tbp]
    \centering
    \begin{footnotesize}
    \begin{equation*}
    \ \ \ \ \ \ 
        \Qcircuit @C=0.4em @R=0.7em @!R {
        \lstick{\ket{x_1}} & \push{\text{\circled{$0$}}} \qw \qwx[1] & \qw &&& \gate{X_{01}} & \qw & \gate{X_{01}} & \push{\text{\circled{$0$}}} \qw \qwx[2] & \gate{X_{01}} & \qw & \gate{X_{01}} \barrier[-1.4em]{3} & \gate{X_{02}} & \push{\text{\circledr{$2$}}} \qw \qwx[3] & \gate{X_{02}} & \gate{X_{01}} & \qw & \gate{X_{01}} & \push{\text{\circled{$0$}}} \qw \qwx[2] & \gate{X_{01}} & \qw & \gate{X_{01}} & \gate{X_{02}} & \push{\text{\circledr{$2$}}} \qw \qwx[3] & \gate{X_{02}} & \qw \\
        \lstick{\ket{x_2}} & \push{\text{\circled{$0$}}} \qw \qwx[1] & \qw &\dstick{
        =}&& \push{\text{\circledr{$1$}}} \qw \qwx[-1] & \gate{X_{01}} & \push{\text{\circledr{$1$}}} \qw \qwx[-1] & \qw & \push{\text{\circledr{$1$}}} \qw \qwx[-1] & \gate{X_{01}} & \push{\text{\circledr{$1$}}} \qw \qwx[-1] & \push{\text{\circled{$0$}}} \qw \qwx[-1] & \qw & \push{\text{\circled{$0$}}} \qw \qwx[-1] & \push{\text{\circledr{$1$}}} \qw \qwx[-1] & \gate{X_{01}} & \push{\text{\circledr{$1$}}} \qw \qwx[-1] & \qw & \push{\text{\circledr{$1$}}} \qw \qwx[-1] & \gate{X_{01}} & \push{\text{\circledr{$1$}}} \qw \qwx[-1] & \push{\text{\circled{$0$}}} \qw \qwx[-1] & \qw & \push{\text{\circled{$0$}}} \qw \qwx[-1] & \qw \\
        \lstick{\ket{t}} & \gate{X_{01}} & \qw &&& \qw & \qw & \qw & \gate{X_{01}} & \qw & \qw & \qw & \qw & \qw & \qw & \qw & \qw & \qw & \gate{X_{01}} & \qw & \qw & \qw & \qw & \qw & \qw & \qw \\
        \lstick{\ket{a}} & \qw & \qw &&& \qw & \push{\text{\circledb{$o$}}} \qw \qwx[-2] & \qw & \qw & \qw & \push{\text{\circledb{$o$}}} \qw \qwx[-2] & \qw & \qw & \gate{X_{eo}^e} & \qw & \qw & \push{\text{\circledb{$o$}}} \qw \qwx[-2] & \qw & \qw & \qw & \push{\text{\circledb{$o$}}} \qw \qwx[-2] & \qw & \qw & \gate{X_{eo}^e} & \qw & \qw
        }
    \end{equation*}
    \end{footnotesize}
    \caption{The synthesis of $\ket{0^2}$-$X_{01}$ for even $d>3$. $\ket{a}$ is a borrowed ancilla.}
    \label{fig:CCX_evend}
\end{figure*}

\begin{theorem}
	\label{the:X01_even}
	For even $d\geq 3$, the $\ket{0^k}$-$X_{01}$ can be synthesized by using $O(k d^3)$ $\mathcal{G}$-gates and one borrowed ancilla.
\end{theorem}

\begin{proof}
	
	\begin{figure}[tbp]
		\centering
		\footnotesize{
		\begin{equation*}
		\Qcircuit @C=0.9em @R=0.5em @!R{
			\lstick{\ket{x_1}} & \qw & \qw & \qw & \push{\text{\circled{$0$}}} \qw \qwx[1] & \qw & \qw & \qw & \qw\\
			\lstick{\ket{x_2}} & \qw & \qw & \qw & \push{\text{\circled{$0$}}} \qw \qwx[1] & \qw & \qw & \qw & \qw\\
			\lstick{\ket{a_1}} & \qw & \qw & \push{\text{\circledb{$o$}}} \qw \qwx[1] & \gate{X_{eo}^{e}} & \push{\text{\circledb{$o$}}} \qw \qwx[1] & \qw & \qw & \qw\\
			\lstick{\ket{x_3}} & \qw & \qw & \push{\text{\circled{$0$}}} \qw \qwx[1] & \qw & \push{\text{\circled{$0$}}} \qw \qwx[1] & \qw & \qw & \qw\\
			\lstick{\ket{a_2}} & \qw & \push{\text{\circledb{$o$}}} \qw \qwx[1] & \gate{X_{eo}^{e}} & \qw & \gate{X_{eo}^{e}} & \push{\text{\circledb{$o$}}} \qw \qwx[1] & \qw & \qw\\
			\lstick{\ket{x_4}} & \qw & \push{\text{\circled{$0$}}} \qw \qwx[1] & \qw & \qw & \qw & \push{\text{\circled{$0$}}} \qw \qwx[1] & \qw & \qw\\
			\lstick{\ket{a_3}} & \push{\text{\circledb{$o$}}} \qw \qwx[1] & \gate{X_{eo}^{e}} & \qw & \qw & \qw & \gate{X_{eo}^{e}} & \push{\text{\circledb{$o$}}} \qw \qwx[1] & \qw \\
			\lstick{\ket{x_5}} & \push{\text{\circled{$0$}}} \qw \qwx[1] & \qw & \qw & \qw & \qw & \qw & \push{\text{\circled{$0$}}} \qw \qwx[1] & \qw \\
			\lstick{\ket{t}} & \gate{X_{01}} & \qw & \qw & \qw & \qw & \qw & \gate{X_{01}} & \qw \gategroup{1}{3}{7}{7}{1.5em}{--}
		}
		\end{equation*}}
		\caption{The synthesis of $\ket{0^5}$-controlled $X_{01}$ gate with 3 garbage ancilla for even $d$. $\ket{\boldsymbol{x}_{1:5}}$ are the control qudits, $\ket{t}$ is the target qudit, $\ket{\boldsymbol{a}_{1:3}}$ are garbage ancilla. Note that the top gate is $\ket{0^2}$-$X_{eo}^e$, the two bottom gates are $\ket{o}\ket{0}$-$X_{01}$, the other gates are $\ket{o}\ket{0}$-$X_{eo}^e$.}
		\label{fig:X01_an_even}
	\end{figure}
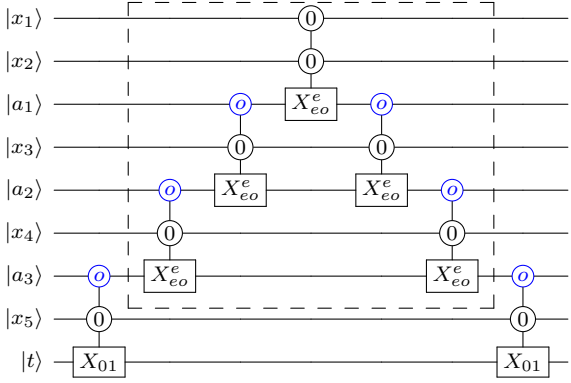
	First, we show how to synthesize the $\ket{0^k}$-$X_{01}$ gate by using $O(k d^3)$ $\mathcal{G}$-gates and $k-2$  {\it garbage} ancilla. The structure of our synthesis is ``$\Lambda$"-like, and Fig.~\ref{fig:X01_an_even} is an illustration for $k=5$. To see the correctness, a key observation is that adding one $\ket{o}\ket{0}$-$X_{eo}^e$ gate to each side of $\ket{0^r}$-$X_{eo}^e$ in the fashion as Fig.~\ref{fig:X01_an_even} produces the $\ket{0^{r+1}}$-$X_{eo}^e$ gate. For example, the circuit in the dashed box implements the $\ket{0^4}$-$X_{eo}^e$ gate whose controls are $\ket{\boldsymbol{x}_{1:4}}$ and target is $\ket{a_3}$. Finally, by adding one $\ket{o}\ket{0}$-$X_{01}$ gate to the $\ket{0^{k-1}}$-$X_{eo}^e$ gate on each side, we obtain the $\ket{0^k}$-$X_{01}$ gate. To verify correctness, it is important to note that the $X_{01}$ gate is applied to $\ket{t}$ only when $x_5=0$ and the circuit within the dashed box changes $\ket{a_3}$. The circuit size can be easily seen to be $O(kd^3)$ by noting that (i) a $\ket{0^2}$-$X_{ij}$ gate can be synthesized by using $O(d)$ $\mathcal{G}$-gates according to Lemma~\ref{lem:2toffoli_even}; (ii) the $\ket{0^2}$-$X_{eo}^e$ gate can be synthesized by using $O(d)$ $\ket{0^2}$-$X_{ij}$ gates; and (iii) the $\ket{o}\ket{0}$-$X_{eo}^e$ gate can be synthesized by using $O(d)$ $\ket{0^2}$-$X_{eo}^e$ gates and $O(d)$ $X_{ij}$ gates.
	
	Second, the initial state of the $k-2$ garbage ancilla can be recovered by adding a component to the end of the circuit which reverses all the gates but the two in the bottom (e.g., adding the reverse of the dashed box to the end in Fig.~\ref{fig:X01_an_even}). So, the $k-2$ garbage ancilla can be replaced with borrowed ancilla while keeping the circuit size still $O(kd^3)$.
	
	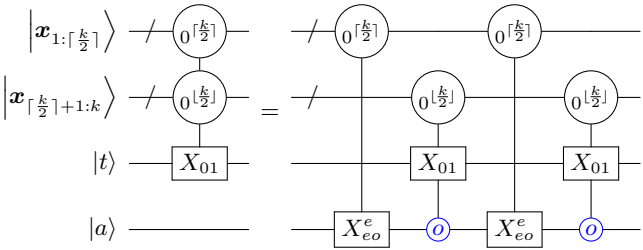
\begin{figure}[tbp]
		\centering
		\begin{small}
		\begin{equation*}
		\ \ \ \ \ \ \ \ \ \ 
		\Qcircuit @C=0.9em @R=0.5em @!R {
			\lstick{\ket{\boldsymbol{x}_{1:\lceil \frac{k}{2}\rceil}}} &{/} \qw & \push{\scriptsize{\text{\circled{$0^{\tiny{\lceil \! \frac{k}{2} \! \rceil}}$}}}} \qw \qwx[1]  & \qw &&& {/} \qw&\push{\scriptsize{\text{\circled{$0^{\tiny{\lceil \! \frac{k}{2} \! \rceil}}$}}}} \qw \qwx[3] & \qw & \push{\scriptsize{\text{\circled{$0^{\tiny{\lceil \! \frac{k}{2} \! \rceil}}$}}}} \qw \qwx[3]& \qw &\qw \\
			\lstick{\ket{\boldsymbol{x}_{\lceil \frac{k}{2}\rceil + 1:k}}}&{/} \qw & \push{\scriptsize{\text{\circled{$0^{\tiny{\lfloor \! \frac{k}{2} \! \rfloor}}$}}}} \qw \qwx[1] & \qw &\dstick{=}&& {/} \qw&\qw & \push{\scriptsize{\text{\circled{$0^{\tiny{\lfloor \! \frac{k}{2} \! \rfloor}}$}}}} \qw \qwx[1] &\qw & \push{\scriptsize{\text{\circled{$0^{\tiny{\lfloor \! \frac{k}{2} \! \rfloor}}$}}}} \qw \qwx[1] &\qw \\
			\lstick{\ket{t}} & \qw & \gate{X_{01}} & \qw &&& \qw &\qw & \gate{X_{01}} &\qw & \gate{X_{01}} &\qw  \\
			\lstick{\ket{a}} & \qw & \qw & \qw &&& \qw &\gate{X_{eo}^{e}} & \push{\text{\circledb{$o$}}} \qw \qwx[-1] &\gate{X_{eo}^{e}} & \push{\text{\circledb{$o$}}} \qw \qwx[-1] &\qw 
		}
		\end{equation*}
		\end{small}
		\caption{The synthesis of $\ket{0^k}$-$X_{01}$ gate with one borrowed ancilla for even $d$. $\ket{\boldsymbol{x}_{1:k}}$ are the control qudits, $\ket{t}$ is the target qudit, $\ket{a}$ is one borrowed ancilla.}
		\label{fig:X+1_1_even}
	\end{figure}

    Finally, by implementing the synthesis illustrated in Fig.~\ref{fig:X+1_1_even}, the number of ancilla can be reduced to only one from the initial $k-2$. {To verify the correctness, one can first see that the ancilla remains unchanged after applying the circuit by noting that $(X_{eo}^e)^2=I$, then one can check that:
    \begin{itemize}
    \item When $\boldsymbol{x}_{1:\lceil k/2\rceil} \neq 0^{\tiny{\lceil k/2 \rceil}}$: $X_{01}^2=I$ pledges that $\ket{t}$ remains unchanged.
    \item When $\boldsymbol{x}_{\lceil k/2\rceil + 1:k}\neq 0^{\tiny{\lfloor k/2 \rfloor}}$: no gates are applied to $\ket{t}$.
    \item When $\boldsymbol{x}_{1:k}=0^k$: only one $\ket{o}\ket{0^{\tiny{\lfloor k/2 \rfloor}}}$-$X_{01}$ is active, which applies the $X_{01}$ gate on $\ket{t}$.
    \end{itemize}}
    Furthermore, to estimate the size of the circuit in Fig.~\ref{fig:X+1_1_even}, we need to synthesize $\ket{0^k}$-$X_{eo}^e$ and $\ket{o}\ket{0^{k-1}}$-$X_{01}$ gates. These two gates can be synthesized by making slight changes to the circuit of $\ket{0^k}$-$X_{01}$, as shown in Fig.~\ref{fig:X01_an_even}. In particular, to synthesize $\ket{0^k}$-$X_{eo}^e$, we can replace the two bottom $\ket{o}\ket{0}$-$X_{01}$ gates with two $\ket{o}\ket{0}$-$X_{eo}^e$ gates. Similarly, to synthesize $\ket{o}\ket{0^{k-1}}$-$X_{01}$, we can replace the top $\ket{00}$-$X_{01}$ gates with the $\ket{o}\ket{0}$-$X_{01}$ gates. Therefore, the circuit size remains $O(kd^3)$.
\end{proof}

We remark that the one borrowed ancilla in Theorem \ref{the:X01_even} is necessary. By contraction, suppose the ancilla can be saved, then in the $k+1$ qudit system ($k$ controls and one target), (i) any $\mathcal{G}$-gate is an even permutation on the computational basis $\{\ket{x_1,\cdots,x_k,t}\mid x_1,\cdots,x_k,t\in[\underline{d}]\}$ and (ii) the $\ket{0^k}$-$X_{01}$ gate is an odd permutation on the computational basis. By the fact that a composition of even permutations is still even, we reach a contradiction. 

\subsection{Synthesis of the $k$-Toffoli gate when $d$ is odd}\label{subsection:32}

For odd $d\geq 3$, Fig.~\ref{fig:CCX} presents an ancilla-free synthesis of the $2$-Toffoli gate by using $O(d)$ $\mathcal{G}$-gates, which is a natural generalization of that for $d=3$ synthesized in \cite{yeh2022constructing}. The $2$-Toffoli gate will be used as a gadget to synthesize $k$-Toffoli for larger $k$.

\begin{lemma}
    \label{lem:2toffoli_odd}
	For odd $d\geq 3$, the $\ket{00}$-$X_{01}$ can be synthesized by using $O(d)$ $\mathcal{G}$-gates.
\end{lemma}
\begin{proof}
    Fig.~\ref{fig:CCX} depicts the synthesis of $2$-Toffoli using $O(d)$ $\mathcal{G}$-gates. {In the rest of the proof, we verify the correctness. First, it is obvious that the two control qudits are unchanged since $X_{+1}X_{-1}=I$. Then, we have%we must establish that the target qudit $\ket{t}$ becomes $X_{01}\ket{t}$ when $x_1=x_2=0$, and remains unchanged otherwise.
%  Assessing all possible input strings:
    \begin{itemize}
    \item If $x_1 \neq 0$, the target $\ket{t}$ is unchanged since $X_{01}^2=I$.
    \item If $x_1 = 0$ and $x_2 \neq 0$, at least one $X_{01}$ gate is applied to the target $\ket{t}$; for odd $x_2$, the left $\ket{e}$-$X_{01}$ gate is activated and for even $x_2$, the right $\ket{e}$-$X_{01}$ gate is activated. Thus, two $X_{01}$ gates are consistently applied, canceling each other out.
    \item If $x_1 = x_2 = 0$, only the leftmost $\ket{0}$-$X_{01}$ gate changes the target qudit to $X_{01}\ket{t}$.
    \end{itemize}
    Therefore, it follows that $\ket{t}$ becomes $X_{01}\ket{t}$ if $x_1 = x_2 = 0$ and unchanged otherwise.}
\end{proof}

\begin{figure}[tbp]
    \centering
    \footnotesize{
    \begin{equation*}
        \Qcircuit @C=0.9em @R=0.7em @!R {
        \lstick{\ket{x_1}}&\push{\text{\circled{$0$}}} \qw \qwx[1] & \qw &&& \push{\text{\circled{$0$}}} \qw \qwx[2] & \push{\text{\circled{$0$}}} \qw \qwx[1] & \qw & \push{\text{\circled{$0$}}} \qw \qwx[1] & \qw & \qw \\
        \lstick{\ket{x_2}}&\push{\text{\circled{$0$}}} \qw \qwx[1] & \qw &\push{=}&& \qw & \gate{X_{+1}} & \push{\text{\circledb{$e$}}} \qw \qwx[1] & \gate{X_{-1}} & \push{\text{\circledb{$e$}}} \qw \qwx[1] & \qw \\
        \lstick{\ket{t}}&\gate{X_{01}} & \qw &&& \gate{X_{01}} & \qw & \gate{X_{01}} & \qw & \gate{X_{01}} & \qw
        }
    \end{equation*}}
    \caption{The synthesis of $\ket{0^2}$-$X_{01}$ gate for odd $d$, where $X_{-1}:=X_{+(d-1)}=X_{+1}^{\dagger}$.}
    \label{fig:CCX}
\end{figure}
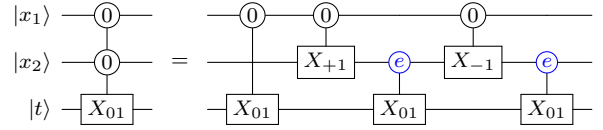

\begin{figure}[tbp]
    \centering
    \footnotesize{
    \begin{equation*}
        \Qcircuit @C=0.9em @R=0.5em @!R {
        & \push{\text{\circledr{$\star$}}} \qw \qwx[1] & \qw &&& \push{\text{\circledr{$1$}}} \qw \qwx[1] & \push{\text{\circledr{$2$}}} \qw \qwx[1] & \qw &\push{\cdots} && \push{\tiny{\text{\circledr{$d\! \! - \! \! 1$}}}} \qw \qwx[1] & \qw\\
        & \push{\text{\circled{$0$}}} \qw \qwx[1] & \qw &\push{=}&& \push{\text{\circled{$0$}}} \qw \qwx[1] & \push{\text{\circled{$0$}}} \qw \qwx[1] & \qw & \push{\cdots} && \push{\text{\circled{$0$}}} \qw \qwx[1] & \qw \\
        & \gate{X_{+\star}} & \qw &&& \gate{X_{+1}} & \gate{X_{+2}} & \qw & \push{\cdots} && \gate{X_{+(d-1)}} & \qw
        }
    \end{equation*}}
    \caption{$\ket{\star}\ket{0}$-$X_{+\star}$ gate.}
    \label{fig:CyX}
\end{figure}

For the convenience of presentation, we define the $\ket{\star}\ket{0}$-$X_{+\star}$ gate as in Fig.~\ref{fig:CyX}, for each state $\ket{y}$ of the first control which implements $X_{+y}$ on the target if the second control is $\ket{0}$. % \zwb{and} $\ket{y}$ is the state of the first control. 
Noting that $\ket{0^2}$-$X_{ij}$ can be synthesized by using $O(d)$ $\mathcal{G}$-gates based on the circuit in Fig.~\ref{fig:CCX}, we have $\ket{0^2}$-$X_{+i}$ can be synthesized by using $O(d^2)$ $\mathcal{G}$-gates for any $i\in[\underline{d}]$, which further implies that $\ket{\star}\ket{0}$-$X_{+\star}$ can be synthesized by using $O(d^3)$ $\mathcal{G}$-gates. 

Similarly, we can define the $\ket{\star}\ket{0}$-$X_{-\star}$ gate, for each state $\ket{y}$ of the first control which implements $X_{-y}:=X_{+(d-y)}$ on the target if the second control is $\ket{0}$. Note that $\ket{\star}\ket{0}$-$X_{-\star}$ can also be synthesized by using $O(d^3)$ $\mathcal{G}$-gates. A useful fact is that $\ket{\star}\ket{0}$-$X_{-\star}$ is the inverse operation of $\ket{\star}\ket{0}$-$X_{+\star}$.

The following synthesis of $\ket{0^k}$-$X_{+1}$ will be used to synthesize $k$-Toffoli gate. 
\begin{lemma}
\label{lem:X+1_1}
For odd $d\geq 3$, $\ket{0^k}$-$X_{+1}$ can be synthesized by using $O(k d^3)$ $\mathcal{G}$-gates and $k-2$ borrowed ancilla.
\end{lemma}

\begin{proof}
	First, we show how to synthesize the $\ket{0^k}$-$X_{+1}$ gate by using $O(k d^3)$ $\mathcal{G}$-gates and $k-2$ {\it garbage} ancilla. The structure of our synthesis is ``$\Lambda$"-like, and Fig.~\ref{fig:X+1_an} is an illustration for $k=5$. To see the correctness, a key observation is that adding one $\ket{\star}\ket{0}$-$X_{-\star}$ gate and one $\ket{\star}\ket{0}$-$X_{+\star}$ gate to the left side and right side respectively of the $\ket{0^r}$-$X_{+1}$ gate in the fashion as Fig.~\ref{fig:X+1_an} produces the $\ket{0^{r+1}}$-$X_{+1}$ gate. For example, the circuit in the dashed box implements $\ket{0^4}$-$X_{+1}$ whose controls are $\ket{\boldsymbol{x}_{1:4}}$ and target is $\ket{a_3}$, then by adding the two gates in the bottom, we obtain the $\ket{0^5}$-$X_{+1}$ gate. The circuit size can be easily seen to be $O(kd^3)$. 
	
	Second, the initial state of the $k-2$ garbage ancilla can be recovered by adding a component to the end of the circuit which reverses all the gates but the two in the bottom (e.g., adding the reverse of the dashed box to the end in Fig.~\ref{fig:X+1_an}). So, the $k-2$ garbage ancilla can be replaced with borrowed ancilla while keeping the circuit size still $O(kd^3)$.
 \end{proof}

\begin{figure}[tbp]
    \centering
    \footnotesize{
    \begin{equation*}
    \Qcircuit @C=0.7em @R=0.5em @!R{
    \lstick{\ket{x_1}} & \qw & \qw & \qw & \push{\text{\circled{$0$}}} \qw \qwx[1] & \qw & \qw & \qw & \qw \\
    \lstick{\ket{x_2}} & \qw & \qw & \qw & \push{\text{\circled{$0$}}} \qw \qwx[1] & \qw & \qw & \qw & \qw\\
    \lstick{\ket{a_1}} & \qw & \qw & \push{\text{\circledr{$\star$}}} \qw \qwx[1] & \gate{X_{+1}} & \push{\text{\circledr{$\star$}}} \qw \qwx[1] & \qw & \qw & \qw\\
    \lstick{\ket{x_3}} & \qw & \qw & \push{\text{\circled{$0$}}} \qw \qwx[1] & \qw & \push{\text{\circled{$0$}}} \qw \qwx[1] & \qw & \qw & \qw\\
    \lstick{\ket{a_2}} & \qw & \push{\text{\circledr{$\star$}}} \qw \qwx[1] & \gate{X_{-\star}} & \qw & \gate{X_{+\star}} & \push{\text{\circledr{$\star$}}} \qw \qwx[1] & \qw & \qw\\
    \lstick{\ket{x_4}} & \qw & \push{\text{\circled{$0$}}} \qw \qwx[1] & \qw & \qw & \qw & \push{\text{\circled{$0$}}} \qw \qwx[1] & \qw & \qw\\
    \lstick{\ket{a_3}} & \push{\text{\circledr{$\star$}}} \qw \qwx[1] & \gate{X_{-\star}} & \qw & \qw & \qw & \gate{X_{+\star}} & \push{\text{\circledr{$\star$}}} \qw \qwx[1] & \qw\\
    \lstick{\ket{x_5}} & \push{\text{\circled{$0$}}} \qw \qwx[1] & \qw & \qw & \qw & \qw & \qw & \push{\text{\circled{$0$}}} \qw \qwx[1] & \qw\\
    \lstick{\ket{t}} & \gate{X_{-\star}} & \qw & \qw & \qw & \qw & \qw & \gate{X_{+\star}} & \qw \gategroup{1}{3}{7}{7}{1.5em}{--}
    }
\end{equation*}}
    \caption{The synthesis of $\ket{0^5}$-$X_{+1}$ gate with 3 garbage ancilla. $\ket{\boldsymbol{x}_{1:5}}$ are the control qudits, $\ket{t}$ is the target qudit, $\ket{\boldsymbol{a}_{1:3}}$ are garbage ancilla.}
    \label{fig:X+1_an}
\end{figure}
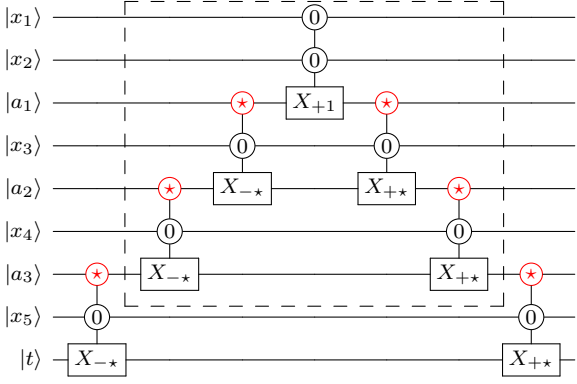

In the rest of this subsection, we show how to synthesize the $k$-Toffoli gate. The framework is presented in Fig.~\ref{fig:X_01}. There, $\ket{0}$-$(X_{eo}^{o})^{\otimes (k-1)}$ represents the composition of $(k-1)$ $\ket{0}$-$X_{eo}^o$ gates whose controls are all $\ket{x_k}$ and targets are $\ket{x_1},\cdots,\ket{x_{k-1}}$ respectively. And $X_{eo}^o$ is defined to be $X_{12}X_{34}\cdots X_{(d-2)(d-1)}$. The module $P_k$ is a classical reversible operation on $k$ qudits and acts as:
\[
P_k \ket{x_1,\dots,x_{k-1},x_k} =\ket{x_1,\dots,x_{k-1},h(x_1,\cdots,x_{k})}
\]
where the function $h:[\underline{d}]^{k}\rightarrow [\underline{d}]$ is defined as follows: given $x_{1},\cdots,x_{k}\in[\underline{d}]$, let $i^\ast\in[k-1]$ be the last index such that $x_{i^\ast}\neq 0$ and $x_{i}=0$ for any $i^\ast < i < k$. If there is no such $i^\ast$, i.e., $x_1\cdots x_{k-1}=0^{k-1}$, we let $i^\ast=\perp$.
\begin{itemize}
	\item $h(x_1,\cdots,x_k)=x_k$ if $i^\ast\neq \perp$ and $x_{i^\ast}$ is odd;
	\item otherwise, $i^\ast=\perp$ or $x_{i^\ast}$ is even, let $h(x_1,\cdots,x_k)=x_k-1 \mod{d}$.
\end{itemize}
For example, when $k=2$, $h(x_1,x_2)=x_{2}$ if $x_1$ is odd, and $h(x_1,x_2)=x_{2}-1 \mod{d}$ otherwise. For general $k$, if $x_1\cdots x_{k-1}=10^{k-2}$, then $i^\ast=1$ and $x_{1}$ is odd, and we have $h(x_1,\cdots,x_{k})=x_k$. We call the qudits $\ket{x_1},\cdots,\ket{x_{k-1}}$ controls of $P_k$ and $\ket{x_{k}}$ target of $P_k$.

\begin{lemma}
\label{lem:Pn}
For odd $d\geq 3$, both $P_k$ and $P_k^\dagger$ can be synthesized by using $O(k d^3)$ $\mathcal{G}$-gates and one borrowed ancilla.
\end{lemma}

\begin{proof}

\begin{figure}[tbp]
    \centering
    \footnotesize{
    \begin{equation*}
    \Qcircuit @C=0.5em @R=0.5em @!R{
        \lstick{\ket{x_1}}& \qw & \qw & \qw & \push{\text{\circled{$0$}}} \qw \qwx[1] & \push{\text{\circledb{$e$}}} \qw \qwx[1] & \qw & \qw & \qw & \qw \\
        \lstick{\ket{a_1}}& \qw & \qw & \push{\text{\circledr{$\star$}}} \qw \qwx[1] & \gate{X_{-1}} & \gate{X_{-1}} & \push{\text{\circledr{$\star$}}} \qw \qwx[1] & \qw & \qw & \qw \\
        \lstick{\ket{x_2}}& \qw & \qw & \push{\text{\circled{$0$}}} \qw \qwx[1] & \push{\text{\circledb{$e$}}} \qw \qwx[1] & \qw & \push{\text{\circled{$0$}}} \qw \qwx[1] & \qw & \qw & \qw  \\
        \lstick{\ket{a_2}}& \qw & \push{\text{\circledr{$\star$}}} \qw \qwx[1] & \gate{X_{-\star}} & \gate{X_{-1}} & \qw & \gate{X_{+\star}} & \push{\text{\circledr{$\star$}}} \qw \qwx[1] & \qw & \qw \\
        \lstick{\ket{x_3}}& \qw & \push{\text{\circled{$0$}}} \qw \qwx[1] & \push{\text{\circledb{$e$}}} \qw \qwx[1] & \qw & \qw & \qw & \push{\text{\circled{$0$}}} \qw \qwx[1] & \qw & \qw \\
        \lstick{\ket{a_3}}& \push{\text{\circledr{$\star$}}} \qw \qwx[1] & \gate{X_{-\star}} & \gate{X_{-1}} & \qw & \qw & \qw & \gate{X_{+\star}} & \push{\text{\circledr{$\star$}}} \qw \qwx[1] & \qw \\
        \lstick{\ket{x_4}}& \push{\text{\circled{$0$}}} \qw \qwx[1] & \push{\text{\circledb{$e$}}} \qw \qwx[1] & \qw & \qw & \qw & \qw & \qw & \push{\text{\circled{$0$}}} \qw \qwx[1] & \qw \\
        \lstick{\ket{x_5}}& \gate{X_{-\star}} & \gate{X_{-1}} & \qw & \qw & \qw & \qw & \qw & \gate{X_{+\star}} & \qw \gategroup{1}{3}{6}{8}{1.5em}{--}
        }
    \end{equation*}}
    \caption{Synthesis of $P_5$ gate with 3 garbage ancilla. $\ket{\boldsymbol{x}_{1:5}}$ are the input qudits, $\ket{\boldsymbol{a}_{1:3}}$ are garbage ancilla.}
    \label{fig:P_n_n-2}
\end{figure}
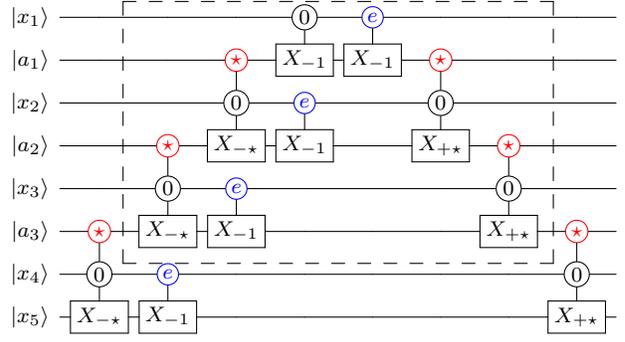
First, we show how to synthesize $P_k$ by using $O(k d^3)$ $\mathcal{G}$-gates and $k-2$ {\it garbage} ancilla. The structure of our synthesis is similar to the synthesis in Lemma~\ref{lem:X+1_1}, and Fig.~\ref{fig:P_n_n-2} is an illustration for $k=5$. The correctness is immediate from the following two observations: 
\begin{enumerate}
	\item The two gates at the top implement $P_2$, whose control is $\ket{x_1}$ and target is $\ket{a_1}$.
	\item Adding one $\ket{\star}\ket{0}$-$X_{-\star}$ and one $\ket{e}$-$X_{-1}$ to the left side of $P_{k-1}$ and one $\ket{\star}\ket{0}$-$X_{+\star}$ to the right side in the fashion as in Fig.~\ref{fig:P_n_n-2}, we get a $P_k$.
	
	For example, the circuit in the dashed box implements $P_{4}$ whose controls are $\ket{\boldsymbol{x}_{1:3}}$ and target is $\ket{a_3}$, then by adding the three gates in the bottom, we obtain the $P_5$.
\end{enumerate}
Besides, the circuit size can be easily seen to be $O(kd^3)$.

Second, the initial state of the $k-2$ garbage ancilla can be recovered by adding a component to the end of the circuit which reverses all the gates but the three at the bottom. So, the $k-2$ garbage ancilla can be replaced with borrowed ancilla while keeping the circuit size still $O(kd^3)$.

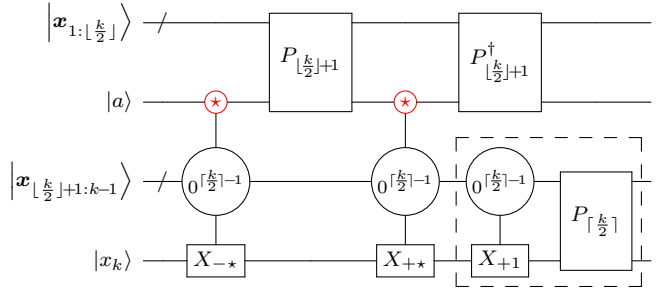
\begin{figure}[tbp]
    \footnotesize{
    \begin{equation*}
    \ \ \ \ \ \ \ \ \ \ \ \ \ \ \ \
    \Qcircuit @C=0.9em @R=0.5em @!R{
        \lstick{\ket{\boldsymbol{x}_{1:\lfloor \frac{k}{2}\rfloor}}}& {/} \qw & \qw & \multigate{1}{P_{\lfloor \! \frac{k}{2} \! \rfloor \! + \! 1}} & \qw & \multigate{1}{P_{\lfloor \! \frac{k}{2} \! \rfloor \! + \! 1}^\dagger} & \qw & \qw \\
        \lstick{\ket{a}}& \qw & \push{\text{\circledr{$\star$}}} \qw \qwx[1] & \ghost{P_{\lfloor \! \frac{k}{2} \! \rfloor \! + \!1}} & \push{\text{\circledr{$\star$}}} \qw \qwx[1] & \ghost{P_{\lfloor \! \frac{k}{2} \! \rfloor \! + \! 1}^\dagger} & \qw & \qw \\
        \lstick{\ket{\boldsymbol{x}_{\lfloor \frac{k}{2}\rfloor \! + \! 1:k \! - \! 1}}}& {/} \qw & \push{\scriptsize{\text{\circled{$0^{\tiny{\lceil \! \frac{k}{2} \! \rceil \! - \! 1}}$}}}} \qw \qwx[1] &\qw & \push{\scriptsize{\text{\circled{$0^{\tiny{\lceil \! \frac{k}{2} \! \rceil \! - \! 1}}$}}}} \qw \qwx[1] &  \push{\scriptsize{\text{\circled{$0^{\tiny{\lceil \! \frac{k}{2} \! \rceil \! - \! 1}}$}}}} \qw \qwx[1]  & \multigate{1}{P_{\lceil \frac{k}{2}\rceil}} & \qw \\
        \lstick{\ket{x_k}}& \qw & \gate{X_{-\star}} & \qw & \gate{X_{+\star}} & \gate{X_{ + 1}} &\ghost{P_{\lceil \frac{k}{2}\rceil}} & \qw \gategroup{3}{6}{4}{7}{0.9em}{--}
        }
    \end{equation*}}
    
    \caption{The synthesis of $P_k$ gate with one borrowed ancilla. $\ket{\boldsymbol{x}_{1:k}}$ are the input qudits, $\ket{a}$ is the ancilla.}
    \label{fig:P_n_half}
\end{figure}

    Finally, via the synthesis in Fig.~\ref{fig:P_n_half}, we can reduce the number of ancilla from $k-2$ to $1$ while keeping the circuit size to $O(kd^3)$. One can easily check the correctness case by case:
    %firstly, the borrowed ancilla $\ket{a}$ stay unchanged for $P_{k} P_{k}^\dagger = I$, then we access all possible input strings:
    \begin{itemize}
        \item When $\boldsymbol{x}_{\lfloor k/2\rfloor + 1:k - 1} \neq 0^{\tiny{\lceil k/2\rceil - 1}}$, $\ket{\star}\ket{0^{\lceil k/2 \rceil -1}}$-$X_{\pm \star}$ and $\ket{0^{\lceil k/2 \rceil -1}}$-$X_{+1}$ gates will not be activated, $P_{\lfloor k/2 \rfloor + 1}$ and $P_{\lfloor k/2 \rfloor + 1}^\dagger$ canceled with each other, thus the whole circuit is equivalent to the $P_{\lceil k/2 \rceil}$ gate. By the definition of $P_k$, $P_k$ equivalence to $P_{\lceil k/2 \rceil}$ if $\ket{\boldsymbol{x}_{\lfloor k/2\rfloor + 1:k - 1}} \neq 0^{\tiny{\lceil k/2 \rceil - 1}}$.
        \item When $\boldsymbol{x}_{\lfloor k/2\rfloor + 1:k - 1} = 0^{\tiny{\lceil k/2 \rceil - 1}}$, the gate in the dashed box canceled with each other, $\ket{\star}\ket{0^{\lceil k/2 \rceil -1}}$-$X_{\pm \star}$ bring what $P_{\lfloor \frac{k}{2} \rfloor + 1}$ do to ancilla $\ket{a}$ to target $\ket{t}$ (a similar technique is used in Fig.~\ref{fig:X+1_an}). By the definition of $P_k$, this exactly is what should do when $\ket{\boldsymbol{x}_{\lfloor k/2\rfloor + 1:k - 1}} = 0^{\tiny{\lceil k/2 \rceil - 1}}$.
    \end{itemize}
    
    To estimate the size of circuit in Fig.~\ref{fig:P_n_half}, the key observation is the $\ket{\star}\ket{0^{\lceil k/2 \rceil -1}}$-$X_{\pm \star}$ gate can be synthesized by slightly changing the circuit for $\ket{0^{\lceil k/2 \rceil}}$-$X_{+1}$ in Fig.~\ref{fig:X+1_an}, where we replace the top $\ket{00}$-$X_{+1}$ gate with $\ket{\star}\ket{0}$-$X_{\pm \star}$ gate. Thus the circuit size is also $O(kd^3)$. In addition, by reversing the synthesis of $P_k$, we get the desired synthesis of $P_k^\dagger$.

\end{proof}

\begin{theorem}
\label{the:X01}
For odd $d\geq 3$, the $k$-Toffoli gate can be synthesized by using $O(kd^3)$ $\mathcal{G}$-gates and no ancilla.
\end{theorem}

\begin{proof}

\begin{figure*}[tbp]
    \centering
    \footnotesize{
    \begin{equation*}
    \Qcircuit @C=1em @R=0.5em @!R{
        \lstick{\ket{\boldsymbol{x}_{1:k-1}}} & {/} \qw & \push{\scriptsize{\text{\circled{$0^{k \! - \! 1}$}}}} \qw \qwx[1] & \qw &&& {/} \qw & \qw & \multigate{1}{P_k} & \qw & \multigate{1}{P_k^{\dag}} & \gate{(X_{eo}^o)^{\otimes (k-1)}} & \multigate{1}{P_k} & \qw & \multigate{1}{P_k^{\dag}} & \gate{(X_{eo}^o)^{\otimes (k-1)}} & \qw \\
        \lstick{\ket{x_k}} & \qw & \push{\text{\circled{$0$}}} \qw \qwx[1] & \qw &\push{=}&& \qw & \push{\text{\circled{$0$}}} \qw \qwx[1] & \ghost{P_k} & \push{\text{\circled{$0$}}} \qw \qwx[1] & \ghost{P_k^{\dag}} & \push{\text{\circled{$0$}}} \qw \qwx[-1] & \ghost{P_k} & \push{\text{\circled{$0$}}} \qw \qwx[1] & \ghost{P_k^{\dag}} & \push{\text{\circled{$0$}}} \qw \qwx[-1] & \qw\\
        \lstick{\ket{t}} & \qw & \gate{X_{01}} & \qw &&& \qw & \gate{X_{01}} & \qw & \gate{X_{01}} & \qw & \qw & \qw & \gate{X_{01}} & \qw & \qw & \qw 
        }
    \end{equation*}}
    \caption{The synthesis of $\ket{0^k}$-$X_{01}$ gate. $\ket{\boldsymbol{x}_{1:k}}$ are the control qudits, $\ket{t}$ is the target qudit.}
    \label{fig:X_01}
\end{figure*}
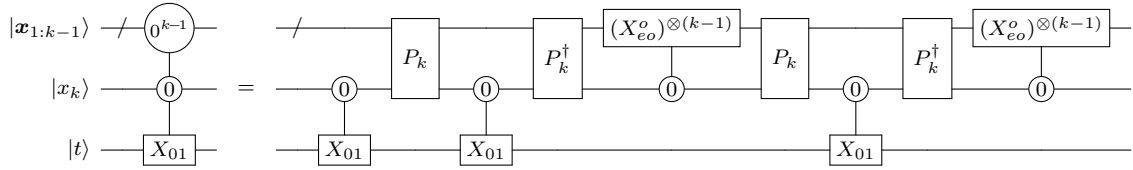

The synthesis is presented in Fig.~\ref{fig:X_01}. By Lemma \ref{lem:Pn}, the circuit size can be easily seen to be $O(kd^3)$. To verify the correctness, we will show that: after implementing the circuit, (i) the controls $\ket{\boldsymbol{x}_{1:k}}$ remain unchanged; (ii) the target $\ket{t}$ becomes $X_{01}\ket{t}$ if $x_1=\cdots=x_{k}=0$ and unchanged otherwise.  

{\it Part (i)}. First, we can delete the three $\ket{0}$-$X_{01}$ gates at the bottom without changing the final state of the controls. Then, noting that the pair of $P_k$ and $P_k^\dagger$ cancels and the pair of two $\ket{0}$-$(X_{eo}^{o})^{\otimes (n-1)}$ gates cancels, we will reach an empty circuit. Now we can conclude that the controls are unchanged.  

{\it Part (ii)}. Given any $x_1\cdots x_{k}\in[\underline{d}]^k$, let $i^\ast\in[k-1]$ be the last index such that $x_{i^\ast}\neq 0$ and $x_{i}=0$ for any $i^\ast < i < k$. If there is no such $i^\ast$, i.e., $x_1\cdots x_{k-1}=0^{k-1}$, we let $i^\ast=\perp$. The definition of $i^\ast$ is the same as that in $P_k$. 

Only the three $\ket{0}$-$X_{01}$ gates at the bottom may change the state of the target $\ket{t}$. The first $\ket{0}$-$X_{01}$ gate fires if and only if $x_{k} = 0$. 

By the definition of $P_k$, the second $\ket{0}$-$X_{01}$ gate fires if and only if one of the following cases happens:
\begin{enumerate}
	\item $x_k = 0$ and  $i^\ast\neq\perp$ and $x_{i^\ast}$ is odd.
	\item $x_k = 1$ and $i^\ast\neq\perp$ and $x_{i^\ast}$ is even.
	\item $x_k=1$ and $i^\ast=\perp$. 
\end{enumerate}

A similar argument to Part (i) shows that the circuit before the first $\ket{0}$-$(X_{eo}^{o})^{\otimes (k-1)}$ gate does not change the state of the controls. Now, let us focus on the circuit after the first $P_k^\dagger$ gate. It is not hard to check that the third $\ket{0}$-$X_{01}$ gate fires if and only if one of the following cases happens:
\begin{enumerate}
	\item $x_k = 0$ and  $i^\ast\neq\perp$ and $x_{i^\ast}$ is even.
	\item $x_k = 1$ and $i^\ast\neq\perp$ and $x_{i^\ast}$ is even.
	\item $x_k=1$ and $i^\ast=\perp$. 
\end{enumerate}

Note that $\ket{t}$ becomes $X_{01}\ket{t}$ if and only if exactly one or exactly three $\ket{0}$-$X_{01}$ gates fire, which is easily seen to happen if and only if $x_k =  0$ and $i^\ast =\perp$, i.e., $x_1=\cdots=x_k=0$.
\end{proof}

\section{Applications}
\label{sec:app}
\subsection{Unitary synthesis on $d$-level qudits}
\label{sec:unitary}
Bullock et al. \cite{bullock2005asymptotically} showed that any unitary on $n$ $d$-level qudits can be exactly synthesized by $O(d^{2n})$ two-qudit gates. The two-qudit gate count has been shown to be optimal \cite{bullock2005asymptotically}. However, as their synthesis uses as many as $\lceil (n-2)/(d-2) \rceil$ clean ancilla, the ancilla count remains to be optimized.

Looking closer at their synthesis, the ancilla are only used in the synthesis of multi-controlled qudit gates. So, by substituting our improved synthesis of multi-controlled qudit gates which uses just one clean ancilla, we can reduce the number of clean ancilla from $\lceil (n-2)/(d-2) \rceil$ to just $1$ while keeping the two-qudit gate count still optimal. 

\begin{theorem}\label{thm:unitary_decompostion}
For any $d\geq 3$, any unitary acting on $n$ $d$-level qudits can be synthesized by using $O(d^{2n})$ two-qudit gates and one clean ancilla.
\end{theorem}

\subsection{Implementation of Classical Reversible Functions}
\label{sec:reversible}
In this subsection, we show how to implement any classical reversible function $f:[\underline{d}]^n\rightarrow [\underline{d}]^n$ on $n$ $d$-level qudits by using $O(nd^n)$ $\mathcal{G}$-gates and one borrowed ancilla (Theorem \ref{thm:classical_ub}). Moreover, for odd $d\geq 3$, this borrowed ancilla can even be saved. In addition, the $\mathcal{G}$-gate count in our implementation can be shown to be optimal up to a $\log n$ factor (Lemma \ref{the:rlower}).

In particular, for $d=3$, each $\mathcal{G}$-gate can be further exactly synthesized by using a constant number of Clifford+T gates \cite{yeh2022constructing}. So Theorem \ref{thm:classical_ub} directly leads to an ancilla-free implementation of ternary classical reversible functions by using $O(n3^n)$ Clifford+T gates, which improves Yeh and van de Wetering's result  \cite{yeh2022constructing}. 

The proof of Theorem \ref{thm:classical_ub} is a natural generalization of that to Theorem 2 in \cite{yeh2022constructing}, except that our improved synthesis of the multi-controlled Toffoli substitutes for that used in \cite{yeh2022constructing}.
\begin{theorem}\label{thm:classical_ub}
Any $n$-variable $d$-ary classical reversible function can be exactly implemented by using $O(n d^{n})$ $\mathcal{G}$-gates. This implementation is ancilla-free when $d$ is odd, and uses one borrowed ancilla when $d$ is even.
\end{theorem}
\begin{proof}
Fix a $n$-variable $d$-ary classical reversible function $f:[\underline{d}]^n \rightarrow  [\underline{d}]^n$. We view $f$ as a permutation on $[\underline{d}]^n$, which can be expressed as product of $(d^n - 1)$ 2-cycles \cite{dixon1996permutation}. We will show that each such 2-cycle can be implemented by using $O(n)$ $\mathcal{G}$-gates, and using no ancilla when $d$ is odd and using one borrowed ancilla when $d$ is even. Then the conclusion is immediate.

In the following, we show how to implement a $2$-cycle. Let $\ket{\boldsymbol{a}}=\ket{a_1\cdots a_n}$ and $\ket{\boldsymbol{b}}=\ket{b_1\cdots b_n}$ be two distinct computational basis of the $n$ $d$-level qudit system. W.l.o.g., we assume $a_n\neq b_n$. Suppose the $2$-cycle swaps $\ket{\boldsymbol{a}}$ and $\ket{\boldsymbol{b}}$ and acts identity on other computational basis states. The implementation of this $2$-cycle is depicted in Fig.~\ref{fig:2cycle}. Here, Step 1 is the composition of $(n-1)$ $\ket{b_n}$-$X_{ij}$ gates whose controls are all $\ket{x_n}$ and targets are $\ket{x_1},\cdots,\ket{x_{n-1}}$ respectively. Step 2 is a multi-controlled $X_{a_nb_n}$ gate, which fires if $\vec{x_{1:n-1}}=\vec{a_{1:n-1}}$. Step 3 is the same as Step 1. By Theorems \ref{the:X01_even} and \ref{the:X01}, this circuit can be synthesized by using $O(n)$ $\mathcal{G}$-gates and no ancilla when $d$ is odd and one borrowed ancilla when $d$ is even. 
To verify the correctness, {one can easily check that:}
%we need to consider all possible input strings:
\begin{itemize}
\item When $x_n \notin \{a_n,b_n\}$, {no gates are activated.}
%no gates can change the state $\ket{\boldsymbol{x}}$.
\item When $\ket{\boldsymbol{x}}=\ket{\boldsymbol{b}}$, the gates on the left and middle are activated, resulting in a transition from state $\ket{\boldsymbol{b}}$ to $\ket{\boldsymbol{a}}$.
\item When $x_n=b_n$ and $\boldsymbol{x}_{1:n-1} \neq \boldsymbol{b}_{1:n-1}$, the gates on the left and right are activated and cancel each other out.
\item When $\ket{\boldsymbol{x}}=\ket{\boldsymbol{a}}$, the gates on the right and middle are activated, resulting in a transition from state $\ket{\boldsymbol{a}}$ to $\ket{\boldsymbol{b}}$.
\item When $x_n=a_n$ and $\boldsymbol{x}_{1:n-1} \neq \boldsymbol{a}_{1:n-1}$, no gates are activated.
\end{itemize}
Thus, the circuit in Fig.~\ref{fig:2cycle} indeed implements the function of 2-cycle, and the size of the circuit can be directly determined from Theorem.\ref{the:X01}.
\end{proof}
 
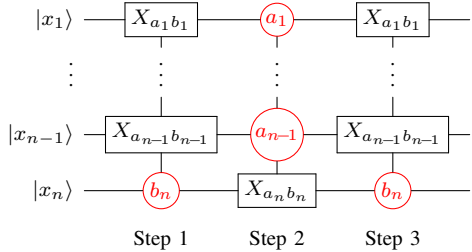
\begin{figure}[tbp]
    \centering
    \footnotesize{
    \begin{equation*}
        \Qcircuit @C=1em @R=0.5em  {
        \lstick{\ket{x_1}} & \gate{X_{a_1 b_1}} \qwx[1] & \push{\text{\circledr{$a_1$}}} \qw \qwx[1] & \gate{X_{a_1 b_1}} \qwx[1] & \qw \\
        \lstick{\vdots}&\push{\vdots}&\push{\vdots}&\push{\vdots}& \\
        &&&&\\
        \lstick{\ket{x_{n-1}}} & \gate{X_{a_{n\! - \! 1} b_{n \! - \! 1}}} \qwx[-1] & \push{\text{\circledr{$a_{n\! - \! 1}$}}} \qw \qwx[1] \qwx[-1] & \gate{X_{a_{n\! - \! 1} b_{n \! - \! 1}}} \qwx[-1] & \qw \\
        \lstick{\ket{x_n}} & \push{\text{\circledr{$b_{n}$}}} \qw \qwx[-1] & \gate{X_{a_n b_n}} & \push{\text{\circledr{$b_{n}$}}} \qw \qwx[-1] & \qw \\
        &\dstick{\mbox{Step 1}}&\dstick{\mbox{Step 2}}&\dstick{\mbox{Step 3}}& 
        }
    \end{equation*}}
    \caption{Circuit that implement the $2$-cycle $(\boldsymbol{a},\boldsymbol{b})$.}
    \label{fig:2cycle}
\end{figure}
Moreover, the following lemma shows that the $\mathcal{G}$-gate count is optimal (up to a logarithmic factor). The proof is similar to that of Proposition 3 in \cite{yeh2022constructing}.
\begin{lemma}
\label{the:rlower}
For $d\geq 2$, there exist $n$-variable $d$-ary classical reversible functions $f$ that require $\Omega(nd^n / \log n)$ $\mathcal{G}$-gates to synthesize, if only $O(n)$ ancilla are available.
\end{lemma}
\begin{proof}
Assume the number of ancilla is $(c-1)n$ for a constant $c$. The number of different $\mathcal{G}$-gates, considering the different positions of gates, is $cn(cn-1)+cnd(d-1)/2$. Additionally, the number of different circuits with $N$ gates is upper-bounded by $(cn(cn-1)+cnd(d-1)/2)^N \leq (cdn)^{2N}$. Note that the number of $n$-variable $d$-ary classical reversible functions $f$ is $(d^n)!$. Therefore, the number of gates N must be large enough such that $(cdn)^{2N} \geq (d^n)!$, which implies that $N \geq \frac{nd^n\log d}{4\log (cdn)}$. Since $c$ and $d$ are constant, we have $N=\Omega(nd^n / \log n)$.
%\begin{align*}
%\\
%2N\log cdn &\geq \log (d^n)! \\
%2N\log cdn &\geq \frac{1}{2}d^n \log d^n \\
%N &\geq \frac{nd^n\log d}{4\log cdn}
%\end{align*}
%Note that $\log(k!) \geq \frac{1}{2}k\log k$, and we have $N=\Omega(nd^n / \log n)$ as $c$ and $d$ are constant.
\end{proof}

\section{Conclusion}
\label{sec:con}
We present a linear-size synthesis scheme for any multi-controlled gate on $d$-level qudits by using just one clean ancilla. This new synthesis leads to an optimal-size and one-clean-ancilla synthesis of unitaries on qudits, and a near-optimal-size and ancilla-free/one-borrowed-ancilla implementation of classical reversible functions as qudit gates.

\bibliographystyle{ieeetr}
\bibliography{reference}

\end{document}